\documentclass[a4paper,12pt]{article}

\usepackage[scale=0.75899]{geometry}
\usepackage{enumitem}
\setlist[itemize]{itemsep=0.5ex, topsep=0pt}
\setlist[description]{itemsep=0.0ex, topsep=0pt}

\usepackage[font={small,it},labelfont=bf]{caption}
\usepackage{subcaption}
\usepackage{chngcntr}
\counterwithin{figure}{section}

\usepackage{hyperref}
\usepackage{amsmath}
\usepackage{amssymb}
\usepackage{amsthm}
\theoremstyle{definition}\newtheorem{dfn}{Definition}[section]
\theoremstyle{plain}\newtheorem{obs}{Observation}[section]
\theoremstyle{plain}\newtheorem{thm}{Theorem}[section]
\theoremstyle{plain}\newtheorem{lem}{Lemma}[section]
\theoremstyle{plain}
\theoremstyle{plain}\newtheorem*{bst}{Bounded States Theorem}
\newcommand\refbst{\hyperref[bst]{Bounded States Theorem}}

\newcommand\state[1]{\textsf{\footnotesize\upshape{}#1}}

\usepackage{xcolor}
\usepackage{setspace}

\usepackage{mathtools}
\usepackage{graphicx}
\usepackage{float}

\usepackage{pbox}
\graphicspath{ {./fig/} }

\usepackage{array}
\newcolumntype{C}{>{$}c<{$}}
\newcolumntype{L}{>{$}l<{$}}
\newcolumntype{R}{>{$}r<{$}}

\newcommand\dMP{\ensuremath{d_\mathrm{MP}}}
\newcommand\chr{{\raisebox{0.5\depth}{$\chi$}}}
\newcommand\ext{{\overline\chr}}
\newcommand\exx{{\widehat\chr}}

\author{Olivier Boes\thanks{Department of Knowledge Engineering, Maastricht University, P.O. Box 616, 6200 MD Maastricht, Netherlands.}, Mareike Fischer\thanks{Institut f\"ur Mathematik und Informatik,
Walther-Rathenau-Stra{\ss}e 47,
17487 Greifswald, Germany.}, Steven Kelk\thanks{Department of Knowledge Engineering, Maastricht University, P.O. Box 616 6200 MD, Maastricht, Netherlands.}}

\title{A linear bound on the number of states in optimal convex characters for
maximum parsimony distance}
\date{}

\begin{document}
\maketitle
\begin{abstract}
\noindent
Given two phylogenetic trees on the same set of taxa $X$, the maximum parsimony
distance $\dMP$ is defined as the maximum, ranging over all characters $\chr$ on $X$,
of the absolute difference in parsimony score induced by $\chr$ on the two
trees. In this note we prove that for binary trees there exists a character achieving this maximum
that is convex on one of the trees (i.e. the parsimony score induced
on that tree is equal to the number of states in the character minus 1) and
such that the number of states in the character is at most $7\dMP-5$. This is the first non-trivial
bound on the number of states required by optimal characters, convex or otherwise. The result potentially
has algorithmic significance because, unlike general characters,
convex characters with a bounded number of states can be enumerated in polynomial time.
\end{abstract}



\section{Introduction}
When phylogenetic trees are inferred from different genes or with different methods, the outcome are often topologically distinct trees, even when the underlying set of species is the same \cite{HusonRuppScornavacca10}. It is natural to ask how different these trees really are, which is why different metrics on phylogenetic trees have been suggested \cite{Kuhner2015}. To name just a few, there is for example the Robinson-Foulds distance \cite{RobinsonFoulds}, as well as  tree rearrangement metrics like the SPR distance or the TBR distance \cite{allen01}. Recently, another metric has been proposed: \emph{maximum parsimony distance} $\dMP$ \cite{dMP-fischer,dMP-moulton}, which is
a lower bound on TBR distance (and thus also SBR distance). Informally this metric consists of finding a character with a low parsimony score on one of the trees and a high parsimony score on the other i.e. it seeks a character that, from a parsimony perspective, distinguishes the most between the two trees. Although the
metric is based on the parsimony score of a tree, which can be computed in polynomial time using e.g. Fitch's algorithm \cite{Fitch}, the metric itself is (like SPR and TBR distance) NP-hard to compute, even on  binary trees \cite{dMP-fischer,dMP-kelk}. The metric also seems extremely difficult
to compute in practice, with exact algorithms based on Integer Linear Programming (ILP) currently limited to trees with 15-20 leaves \cite{dMP-kelk}.  

In \cite{dMP-fischer,dMP-moulton} it has been shown that, with a view towards developing more efficient exponential-time algorithms, the search for optimal characters can be restricted to characters which are convex (equivalently, homoplasy-free \cite{dMP-moulton}) on one of the two trees under investigation i.e. the parsimony score on that tree is the number of states in the character minus 1. This immediately yields a trivial algorithm with running time $\text{O}( 4^{n} \cdot \text{poly}(n) )$, where $n$
is the number of leaves in the trees: guess which tree is convex, and then guess the subset of the $\text{O}( 2n )$ edges in this convex tree where mutations occur.  This leads naturally to the question: 
if $\dMP$ is bounded (i.e. ``small''), is it sufficient to restrict our search to convex characters with a bounded number of states (i.e. to locating bounded-size subsets of mutation edges in the convex tree), irrespective of the number of leaves $n$ in the trees? Such questions
are pertinent to the development of fixed parameter tractable algorithms i.e. algorithms that run quickly on trees with a large number of leaves as long as the distance is small (see e.g. \cite{whidden2013fixed} for related discussions). Prior to this note the best bound on the number of states required was $\lfloor n/2 \rfloor$ \cite{dMP-fischer,dMP-kelk}. Here
we show that the number of states required can indeed be decoupled from $n$. In
particular we show that optimal convex characters
exist with at most $7\dMP-5$ states, which is sharp for $\dMP=1$.
 
We conclude with a discussion of the rather subtle complexity consequences of this result, and whether there is room to tighten the bound further.

\section{Preliminaries}


An unrooted binary phylogenetic $X$-tree $T$ is a tree with only vertices of degree 1 (leaves) or 3 (inner vertices) such that the leaves are bijectively labeled by some finite label set $X$ (where $X$ is often called the set of \emph{taxa}). For brevity, such a tree will simply be called \emph{$X$-tree} in the following. A \emph{character} on $X$ is a surjective map $\chr : X \to \mathcal{C}$
where $\mathcal{C}$ is a set of character \emph{states};
the number of distinct states in the character is denoted by $|\chr|$.
An \emph{extension} $\ext$ of a character $\chr$ to a whole $X$-tree $T$
is a map $\ext : \mathcal{V}(T) \to \mathcal{C}$ such that
$\ext(x) = \chr(x)$ for all $x \in X$.
A \emph{mutation}
 induced by $\ext$ in $T$ is an edge
$\{u,v\} \in \mathcal{E}(T)$ satisfying $\ext(u) \neq \ext(v)$,
and we write $\Delta(T,\ext)$ for the set of all mutation edges.
The extension $\ext$ is said to be \emph{most parsimonious}
if it achieves the minimum number of mutations
over all possible extensions to $T$ of the character $\chr$. This leads naturally
to the definition of \emph{parsimony score}.

\begin{dfn}
\label{dfn:len}
    Let $T$ be any $X$-tree and let $\chr$ be any character on $X$. \\
    Then the \emph{parsimony score} of $\chr$ on $T$ is
    \[
        \ell(T, \chr)
        ~:=~ \min_\ext\,\left|\Delta(T,\ext)\right|
        ~=~ \min_\ext\,\left|\left\{~
            \{u,v\} \in \mathcal{E}(T) ~\mid~ \ext(u) \neq \ext(v) 
        ~\right\}\right|
    \]
    where the minimum is taken over all possible extensions $\ext$
    of the character $\chr$ to $T$.
\end{dfn}


It is well-known that $\,\ell(T, \chr) \geq |\chr| - 1\,$.
%
When a character $\chr$ achieves this $\,\ell(T, \chr) = |\chr|-1\,$ minimum,
then $\chr$ is said to be a \emph{convex} character on $T$.
Some authors follow a slightly different (but equivalent) path, by defining the
\emph{homoplasy score} $\,h(T, \chr) := \ell(T, \chr) - |\chr| + 1\,$ of
a character $\chr$ on $T$ \cite{dMP-moulton}.
In this terminology, we have $\,h(T,\chr) \geq 0\,$ and a character $\chr$
attaining the $\,h(T,\chr)=0\,$ minimum is said to be \emph{homoplasy-free}
(with respect to $T$). Clearly, a character is convex if and only if it is homoplasy-free.

Although characters are defined on a set $X$ of taxa, this set of taxa will
often be made implicit, allowing us to speak of a character on an $X$-tree.
We now use the parsimony score to define a distance function on pairs of
$X$-trees.

\begin{dfn}
\label{dfn:dMP}
    Let $(T_1, T_2)$ be a pair of $X$-trees. \\
    Then the \emph{maximum parsimony distance} between $T_1$ and $T_2$ is
    \[ 
        \dMP(T_1,T_2) ~:=~
        \max_\chr\,\left|\,\ell(T_1, \chr) - \ell(T_2, \chr)\,\right| 
    \]
    where the maximum is taken over all possible characters $\chr$ on $X$.
\end{dfn}

It is known that $\dMP$ is a metric on unrooted phylogenetic trees \cite{dMP-fischer}, hence
we call it a distance. However it is not a metric on \emph{rooted} phylogenetic
trees, because then we lose identity of indiscernibles (i.e. we only get
a pseudometric).

A character $\chr$ on a set $X$ of taxa is said to
\emph{achieve distance $k$} on a pair $(T_1,T_2)$ of $X$-trees when
$\left| \ell(T_1, \chr) - \ell(T_2, \chr) \right| = k$.
If this character achieves distance $\dMP(T_1, T_2)$, then we say
that $\chr$ is an \emph{optimal} character for this pair of trees.

An optimal character for a pair of trees which has the additional property
of being convex on at least one of the trees is (predictably) called
an \emph{optimal convex} character (for this pair of trees). 

\section{Result}

We recall the following earlier result, proven in
\cite[Theorem 3.6]{dMP-fischer} and \cite[Observation 6.1]{dMP-kelk}:

\begin{thm}
\label{thm:optconv}
\emph{\cite{dMP-fischer,dMP-kelk}}    Any pair $(T_1,T_2)$ of $X$-trees admits an optimal convex character with at most $\lfloor |X| / 2 \rfloor$ states.
\end{thm}

Our main result is the following new bound which is independent of $|X|$. This is particularly
advantageous when $\dMP$ is small and $|X|$ is large.

\begin{bst}
\label{bst}
Any pair $(T_1,T_2)$ of $X$-trees admits an optimal convex character with at most $7 \cdot \dMP\left(T_1,T_2\right) - 5$
states.
\end{bst}

We will prove this theorem subsequently, but first we need to introduce some more concepts and lemmas in the following two sections.

\subsection{The forest induced by a character extension}

In this section we define the forest $F$ \emph{induced} by an extension $\ext$
(of a character $\chr$ to a $X$-tree $T$); this construction will be
extensively used in the proof of the \refbst.

Let us assume that $\ext$ creates $(p\!-\!1)$ mutations in $T$.
If we delete all these mutation edges, we are left with a forest $F$ having
$p$ connected components. Each of these components is a subtree of $T$, whose
vertices all share a common character state (assigned by $\ext$).
We then say that two components of $F$ are \emph{adjacent} if the
two corresponding subtrees of $T$ are connected by one mutation edge
(they cannot be connected by more than one mutation edge, since there are no cycles in $T$).
This yields a graph structure $G(F)$ where the vertices are the components
of $F$ and the edges are the (unordered) pairs of adjacent components,
which can be identified with the mutation edges of $T$.
$G(F)$ has $p$ vertices and $(p-1)$ edges, and must be connected
since $T$ is connected: therefore $G(F)$ can be seen as a tree in its own right.
Figure~\ref{fig:forest} gives a concrete example of such an induced forest.

\begin{figure}[H]
    \centering
    \begin{subfigure}[c]{0.4\textwidth}
        \centering
        \includegraphics{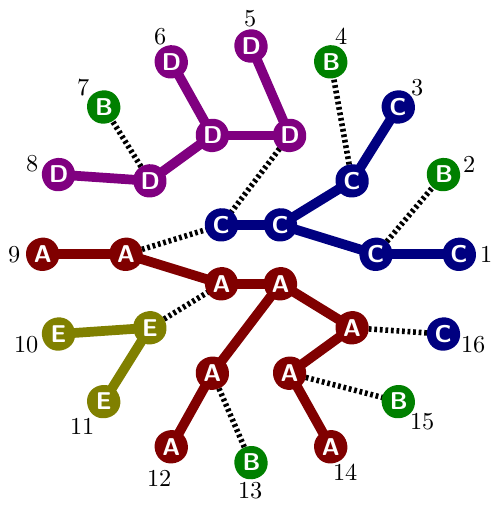}
        \caption{The forest $F$.}
    \end{subfigure}\begin{subfigure}[c]{0.4\textwidth}
        \centering
        \includegraphics{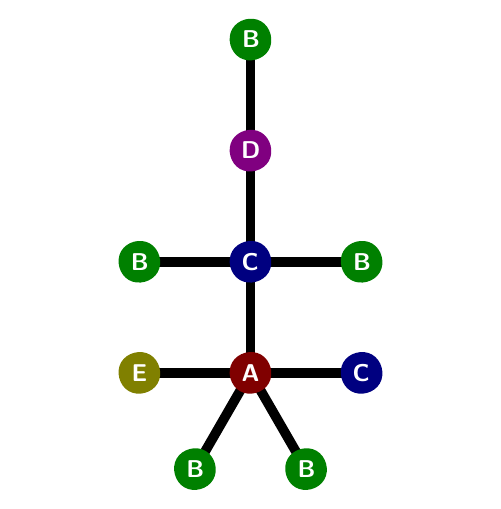}
        \caption{The graph $G(F)$.}
    \end{subfigure}
    \caption[forest]{The forest $F$ induced by a most parsimonious extension
    $\ext$ of the character $\chr = (\state{CBCBDDBDAEEABABC})$ on an $X$-tree
    with leaves labeled from $1$ to $16$, along with its graph structure
    $G(F)$. States $\state{B}$ and $\state{C}$ are repeating states, while all
    others are unique states.}
    \label{fig:forest}
\end{figure}

When $\ext$ is a most parsimonious extension, each component
of the forest must contain at least one leaf of $T$. This in turn implies that a most parsimonious extension never
introduces redundant states i.e. states that were not in the original character. Also,  keep in mind that the forest (and its tree structure) depends on the
choice of the extension $\ext$: even two different most parsimonious
extensions may yield different induced forests.
We conclude this section with some useful terminology and related lemmas.

\begin{dfn}
\label{dfn:states}

    Let $F$ be the forest induced by a most parsimonious extension $\ext$
    of a character $\chr$. Let $\mathcal{C}$ be the set of states used by $\ext$ (which will be equal to the set of states used by $\chr$). We can distinguish between different kinds of states and components:
    \begin{itemize}
        \item
            a state of $\chr$ is \emph{unique} if it is
            assigned to exactly one component of $F$,
        \item
            a state of $\chr$ is \emph{repeating} if it is
            assigned to at least two components of $F$,
        \item
            a component of $F$ is \emph{unique} if its
            assigned state is an unique state of $\chr$,
        \item
            a component of $F$ is \emph{repeating} if its
            assigned state is a repeating state of $\chr$.
    \end{itemize}
    Note that each state is either unique or repeating, but not both.
\end{dfn}

The following lemma gives useful bounds on the numbers of
unique or repeating states and components for a given induced forest.


\begin{lem}
\label{lem:states}
    Let $F$ be the forest induced by any most parsimonious extension $\ext$ of any character
    $\,\chr:X\to\mathcal{C}\,$ to any $X$-tree $T$. The total number of
    components in $F$ is
    $\,|\chr|+h\, =  \ell(T, \chr) + 1$, where $\,h := h(T, \chr)\,$ is the homoplasy score of
    $\chr$ on $T$.
    Then the following inequalities are satisfied.
    \[\begin{array}{rcRCLcl}
        |\chr| - h &~\leq & number of & unique    & states     &~\leq & |\chr|
\\               0 &~\leq & number of & repeating & states     &~\leq & h
\\      |\chr| - h &~\leq & number of & unique    & components &~\leq & |\chr|
\\               h &~\leq & number of & repeating & components &~\leq & 2\,h
    \end{array}
    \]
    Furthermore, $\;\chr$ is convex $\;\Leftrightarrow\;$ $h=0$
    $\;\Leftrightarrow\;$ all states and components are unique.
\end{lem}
\begin{proof} Let us partition $\mathcal{C}$ into two sets
$\mathcal{C}_\mathrm{U}$ and $\mathcal{C}_\mathrm{R}$,
respectively containing the unique states and the repeating states.
The set of components in $F$ is similarly split into two sets
$F_\mathrm{U}$ and $F_\mathrm{R}$.
Clearly, we have:~
$|\mathcal{C}_\mathrm{U}| + |\mathcal{C}_\mathrm{R}| \;=\; |\chr|$
~and~ $|F_\mathrm{U}| + |F_\mathrm{R}| \;=\; |\chr| + h$.

Now, according to Definition~\ref{dfn:states} a state is repeating if it is
assigned to at least two (repeating) components of $F$, and every
component has exactly one state assigned to it, so we must have
$2\,|\mathcal{C}_\mathrm{R}| \leq |F_\mathrm{R}|$.
It is also clear that $|\mathcal{C}_\mathrm{U}|=|F_\mathrm{U}|$, because there
is a one-to-one correspondence between unique states and unique components.
Using these two observations and the two preceding equalities, we find:
\[\begin{array}{lccrcccc}
    &|\mathcal{C}_\mathrm{U}| &+& 2\,|\mathcal{C}_\mathrm{R}|
    &~\leq~&
    |F_\mathrm{U}| &+& |F_\mathrm{R}|
\\[1ex] \Longrightarrow\quad~
    &|\chr| &+& |\mathcal{C}_\mathrm{R}|
    &~\leq~&
    |\chr| &+& h
\end{array}\]

Then canceling the $|\chr|$ term in both sides and combining with the obvious
$\,0 \leq |\mathcal{C}_\mathrm{R}|\,$ bound gives the second inequality of
the lemma, which in turn lead to all three others:
\[\begin{array}{lccccc>{\itshape}L}
    & 0 & ~\leq~ & |\mathcal{C}_\mathrm{R} | & ~\leq~ & h
& (2nd inequality) \\[1ex] \Longrightarrow\quad~
    & 0 & ~\leq~ & |\chr| \;-\; |\mathcal{C}_\mathrm{U} | & ~\leq~ & h
& \\[1ex] \Longrightarrow\quad~
    & -h & ~\leq~ & |\mathcal{C}_\mathrm{U} | \;-\; |\chr| & ~\leq~ & 0
& \\[1ex] \Longrightarrow\quad~
    & |\chr|-h & ~\leq~ & |\mathcal{C}_\mathrm{U} | & ~\leq~ & |\chr|
& (1st inequality) \\[1ex] \Longrightarrow\quad~
    & |\chr|-h & ~\leq~ & |F_\mathrm{U} | & ~\leq~ & |\chr|
& (3rd inequality) \\[1ex] \Longrightarrow\quad~
    & |\chr|-h & ~\leq~ & |\chr| \;+\; h \;-\; |F_\mathrm{R}| & ~\leq~ & |\chr|
& \\[1ex] \Longrightarrow\quad~
    & -|\chr| & ~\leq~ & |F_\mathrm{R}| \;-\; |\chr| \;-\; h & ~\leq~ & h-|\chr|
& \\[1ex] \Longrightarrow\quad~
    & h & ~\leq~ & |F_\mathrm{R}| & ~\leq~ & 2\,h
& (4th inequality)
\end{array}\]

Moreover, if $h=0$, with the 1st inequality we get $ |\mathcal{C}_\mathrm{U} |=|\chi|$, and with the 3rd inequality we get $ |F_\mathrm{U} |=|\chi|$, which implies that all states and all components are unique. On the other hand, if all states and components are unique, we have $|F_\mathrm{R}|=0$, which leads to $h=0$ by the 4th inequality. This completes the proof.
\end{proof}

\subsection{Relabeling states and sufficient conditions for the existence of ``good'' pairs of states}

Here \emph{relabeling} the states
of a given character $\,\chr : X \to \mathcal{C}\,$ simply means composing it
with some surjection $\,\varphi : \mathcal{C} \to \mathcal{C}'\,$ in order to
produce a new character $\,\chr' := \varphi \circ \chr : X \to \mathcal{C}'\,$.
Clearly, $|\chr'| \leq |\chr|$ and
$\,\ell(T,\chr') \leq \ell(T,\chr)\,$ for every $X$-tree $T$.
The proof of the {\refbst}  is based on a relabeling argument in which
only one state of the character is relabeled, i.e.
when $\varphi(\state{A}) = \state{B}$ for two states
$\state{A}, \state{B} \in \mathcal{C}$ but $\varphi$ stays the identity on
states other than $\state{A}$. The high-level idea is to show that, whenever an optimal convex
character exists with more than $7\dMP\left(T_1,T_2\right) - 5$ states, it will always be possible
to find two states $\state{A}$ and $\state{B}$ such that relabeling $\state{A}$ as $\state{B}$ causes the parsimony
score of both trees to decrease by exactly one. That is, a new optimal convex character with \emph{fewer} states can be found, and
the theorem will follow.

 Let $(T_1, T_2)$ be a pair of $X$-trees and let $\chi$ be an \emph{optimal} convex character for this pair. Without loss of generality, let $\chi$ be convex on $T_1$.  Let $\ext_1$ be a most parsimonious extension
of $\chi$ to $T_1$ and $\ext_2$ a most parsimonious extension of $\chi$ to $T_2$. Let $F_1$ and $F_2$ be the forests induced by $\ext_1$ and $\ext_2$
respectively. We say that two components $\state{A}$ and $\state{B}$ are $F_i$-\emph{adjacent} if they are adjacent in the forest $F_i$.  (Note that if a state
is unique, or we are focussing on $F_1$, the term ``state'' and ``component'' can be used interchangeably.)

\begin{obs}
\label{obs:adjacent}
Let $\state{A}$ and $\state{B}$ be two distinct states that are $F_1$-adjacent. Let $\chi'$ be the new character obtained by relabeling $\state{A} := \state{B}$. Then $\chi'$ is a convex character.
In particular,
$ \ell(T_1, \chr') = \ell(T_1, \chr) - 1$ and $\chr'$ uses exactly one fewer state than $\chr$.
Moreover, if $\ell(T_2,\chr') \geq \ell(T_2,\chr) - 1$, then $\chi'$ is an \emph{optimal} convex character (that uses exactly one fewer state than $\chi$).
\end{obs}
\begin{proof}
Relabeling $\state{A} := \state{B}$ within the extension $\ext_1$ yields an extension $\ext'_1$ (of $\chi'$) such that $|\Delta(T_1,\ext'_1)| \leq |\Delta(T_1,\ext_1)| - 1$. This is because a mutation is saved
on the edge generating the adjacency between $\state{A}$ and $\state{B}$.  Hence, $\ell(T_1, \chr' ) \leq  \ell(T_1, \chr) - 1$. Given that $|\chi'| = |\chi|-1$, and the natural lower bound
$\ell(T_1, \chr' ) \geq |\chr'|-1$, it follows that $\ell(T_1, \chr' ) \geq |\chr'|-1=|\chi|-2=  \ell(T_1, \chr) - 1$, and
the convexity of $\chr'$ follows. If, additionally, $\ell(T_2,\chr') \geq \ell(T_2,\chr) - 1$ then the optimality of $\chr'$ is immediate.
\end{proof}

We are thus interested in identifying states $\state{A}$ and $\state{B}$ with the following property: $\state{A}$ and $\state{B}$  are $F_1$-adjacent and
$\ell(T_2,\chr') \geq \ell(T_2,\chr) - 1$ where $\chr'$ is obtained by taking $\state{A} := \state{B}$. We call such a pair of states a \emph{good pair}.

Given an $X$-tree $T$
and an edge $e$ of $T$, deleting $e$ breaks $T$ into two connected components and this naturally induces a bipartition $P|Q$ of $X$. We say then
that $P|Q$ is the \emph{split generated in $T$ by $e$}.

\begin{lem}
\label{lem:split}
Let $\state{A}$ and $\state{B}$ be two distinct states that are $F_1$-adjacent and let $X_{\state{A}}, X_{\state{B}} \subseteq X$ be the taxa that
are labeled with $\state{A}, \state{B}$ respectively. Suppose that in $T_2$, there exists an edge $e$ that generates a split $P|Q$, where $X_{\state{A}} \subseteq P$ and $X_{\state{B}} \subseteq Q$.
Then $(\state{A}, \state{B})$ is a good pair. 
\end{lem}
\begin{proof}
It is sufficient to prove $\ell(T_2,\chr') \geq \ell(T_2,\chr) - 1$. Suppose, for the sake of contradiction, $\ell(T_2,\chr') \leq \ell(T_2,\chr) - 2$. Let $\ext'_2$ be
a most parsimonious extension of $\chi'$ to $T_2$. Deleting $e$ from $T_2$ breaks $\mathcal{V}(T_2)$ into two connected components $\mathcal{V}_{\state{A}}$ and $\mathcal{V}_{\state{B}}$, one containing all taxa  $X_{\state{A}}$ and the other containing all
taxa $X_{\state{B}}$. (Note that here  $X_{\state{A}},  X_{\state{B}}$ refer to the taxa that were labeled $\state{A}$ and $\state{B}$ \emph{before} the relabeling). We  adjust $\ext'_2$ as follows: every vertex
that is in $\mathcal{V}_{\state{A}}$ and labeled with state \state{B}, is switched to state \state{A}. This yields an extension $\exx$ of $\chi$ to $T_2$ such that
$|\Delta(T_2,\exx)| \leq |\Delta(T_2,\ext'_2)| + 1$. This is because the only new mutation that can be created is on the edge $e$. However, this implies $|\Delta(T_2,\exx)|  \leq |\Delta(T_2,\ext'_2)| + 1 \leq \ell(T_2,\chr') +1  \leq (\ell(T_2,\chr) - 2) +1 < \ell(T_2,\chr)$, yielding
a contradiction.  
\end{proof}

 Recall the definitions of \emph{unique} and \emph{repeating} from earlier. We emphasise that here we classify states as unique or repeating with reference to $F_2$ (which is induced by $\ext_2$).

\begin{obs}
\label{obs:pendant}
Let $\state{A}$ and $\state{B}$ be two distinct states that are $F_1$-adjacent and let $\state{A}$ be a unique state.  Let $X_{\state{A}}, X_{\state{B}} \subseteq X$ be the taxa that
are labeled with $\state{A}, \state{B}$ respectively. Suppose that in $T_2$, there exists an edge $e$ that generates a split $X_{\state{A}}|Q$ (i.e. the $X_{\state{A}}$ taxa form a ``pendant subtree'' in $T_2$). Then $(\state{A}, \state{B})$ is a good pair.
\end{obs}

\begin{obs}
\label{obs:twounique}
Let $\state{A}$ and $\state{B}$ be two distinct states that are $F_1$-adjacent and such that both are unique.Then $(\state{A}, \state{B})$ is a good pair.
\end{obs}

\begin{proof}
Observation \ref{obs:pendant} is immediate from Lemma \ref{lem:split}. Observation \ref{obs:twounique} is slightly more subtle. The point here is that if a state $\state{U}$ is unique then in $T_2$ all the vertices allocated state $\state{U}$ (by extension $\ext_2$) form a \emph{single} connected subgraph. In particular this applies to both $\state{A}$ and $\state{B}$. Given that these two states are necessarily distinct, any simple path in $T_2$  between these two connected subgraphs must pass through some edge in $\Delta(T_2,\ext_2)$, and this edge generates a split with all the $\state{A}$ taxa on one side
and all the $\state{B}$ taxa on the other, so Lemma \ref{lem:split} applies.
\end{proof}

See figure~\ref{fig:obs12} for an example where Observations
\ref{obs:pendant} and \ref{obs:twounique} may be used to decrease the
number of character states.

\begin{figure}[H]
    \centering
    \begin{subfigure}[c]{0.33\textwidth}
        \centering
        \includegraphics{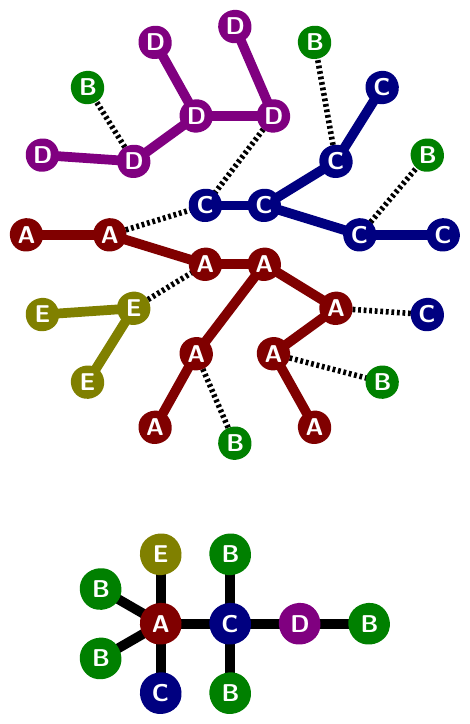}
        \caption{Before any relabeling.\\
                 $\chr = (\state{CBCBDDBDAEEABABC})$}
    \end{subfigure}\begin{subfigure}[c]{0.33\textwidth}
        \centering
        \includegraphics{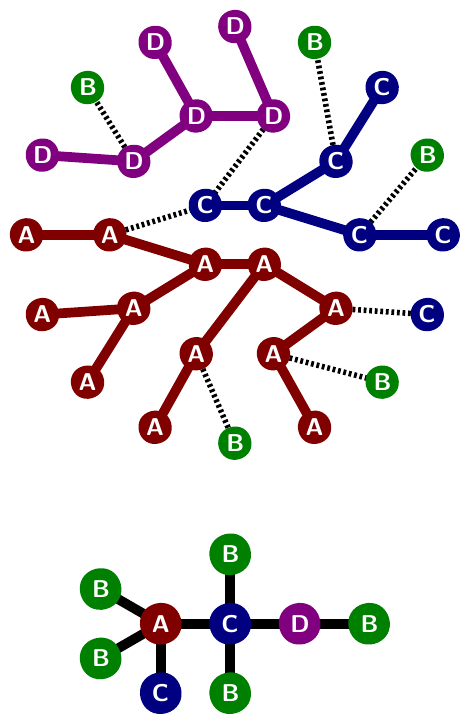}
        \caption{After relabeling $\state{E} := \state{A}$.\\
                 $\chr' = (\state{CBCBDDBDAAAABABC})$}
    \end{subfigure}\begin{subfigure}[c]{0.33\textwidth}
        \centering
        \includegraphics{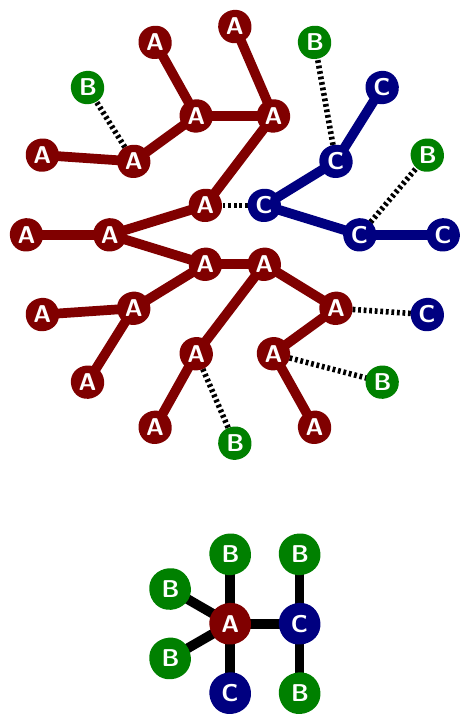}
        \caption{After relabeling $\state{D} := \state{A}$.\\
                 $\chr'' = (\state{CBCBAABAAAAABABC})$}
    \end{subfigure}
    \captionsetup{singlelinecheck=off}
    \caption[obs12]{Successive applications of Observations \ref{obs:pendant}
    and \ref{obs:twounique} to decrease the number of states used by
    an optimal convex character.
    Only the second forests ($F_2$ and its subsequent transformations), along
    with their corresponding graph structures, are shown in these figures.

    \begin{description}
        \item[\quad\itshape(a)]
        The original $F_2$ forest before any relabeling of the states
        of the $\chr$ character.
        The state $\state{E}$ is unique and its component
        in $F_2$ is a pendant subtree.
        Assuming that $\state{E}$ is $F_1$-adjacent to $\state{A}$,
        Observation \ref{obs:pendant} applies and we may relabel
        $\state{E} := \state{A}$. This gives a new optimal convex character
        $\chr'$ which does not use the state $\state{E}$ anymore.
        \item[\quad\itshape(b)]
        The forest $F'_2$ induced by a most parsimonious extension $\ext'_2$
        of $\chr'$ to $T_2$ (note that this is not the only possibility:
        another $\ext'_2$ could induce another $F'_2$).
        States $\state{A}$ and $\state{D}$ are both unique in $F'_2$.
        Assuming $F'_1$-adjacency (where $F'_1$ is induced by some $\ext'_1$),
        Observation \ref{obs:twounique} applies and we may relabel
        $\state{D} := \state{A}$.
        This gives yet another optimal convex character $\chr''$.
        \item[\quad\itshape(c)]
        The forest $F''_2$ induced by a most parsimonious extension
        $\ext''_2$ of $\chr''$ to $T_2$.
        Only three states $\state{A}$, $\state{B}$, and $\state{C}$
        are used by $\chr''$, compared to five states in the original $\chr$
        character.
    \end{description}}
    \label{fig:obs12}
\end{figure}

\begin{lem}
\label{lem:neighbours}
Let $\state{A}$ and $\state{B}$ be two distinct states that are $F_1$-adjacent where $\state{A}$ is a unique state. Assume the situation described in Observation \ref{obs:pendant} does not hold, i.e. there is no edge $e$ which generates a split $X_\state{A}|*$ in $T_2$. If there exists a unique state $\state{C} \neq \state{A}$ such that $\state{A}$ and $\state{C}$ are $F_2$-adjacent and both of degree 2 in $G(F_2)$, then $(\state{A},\state{B})$ is a good pair.
\end{lem}
\begin{proof}
If $\state{A}$ and $\state{B}$ are both unique then we are done, by Observation \ref{obs:twounique}. Hence we may assume that $\state{B}$ is a repeating state i.e. there are at least 2 components in $F_2$ that have state $\state{B}$. Let $\mathcal{V}_{\state{A}}, \mathcal{V}_{\state{C}} \subseteq \mathcal{V}(T_2)$ be those vertices of $T_2$ that are allocated state \state{A}, \state{C} (respectively) by $\ext_2$. Let $X_{\state{A}}, X_{\state{B}}, X_{\state{C}} \subseteq X$ be defined similarly for taxa. We have $|X_{\state{A}}|, |X_{\state{C}}| \geq 2$ because otherwise the situation in Observation \ref{obs:pendant} would trivially apply.

Let $e_{\state{A}\state{C}} \in \Delta(T_2,\ext_2)$ be the edge of $T_2$ that defines the adjacency between $\state{A}$ and $\state{C}$ in $F_2$. 
Let $e_{\state{A}}  \in \Delta(T_2,\ext_2)$ be the edge of $T_2$ that defines the adjacency between $\state{A}$ and its \emph{other} neighbouring component in $F_2$.   Define $e_{\state{C}}$ correspondingly for state $\state{C}$. These three edges are uniquely defined and have no endpoints in common. This is because of the assumption that
Observation \ref{obs:pendant} does not apply, the fact that $T_2$ is a binary tree, and the degree 2 restriction. See figure \ref{fig:deg2mutation} (top subfigure) for a schematic
depiction of the situation.

 Observe that, if $P$ is any simple path (in $T_2$) from a taxon in $X_{\state{A}}$ to a taxon in
$X_{\state{B}}$, then exactly one of the following two situations must hold: (1) $P$  traverses edge $e_{\state{A}}$; (2) $P$ traverses both edges $e_{\state{AC}}$ and $e_{\state{C}}$. This, again, is a consequence of
the degree 2 assumption. We will use this insight in due course.

As usual let $\chi'$ be the character obtained by relabeling $\state{A} := \state{B}$ within $\chi$. (We emphasize that $\mathcal{V}_{\state{A}}, \mathcal{V}_{\state{C}}, X_{\state{A}},X_{\state{B}}, X_{\state{C}}$ are defined \emph{before} the relabeling.) Assume, again for the sake of contradiction, that $\ell(T_2,\chr') \leq \ell(T_2,\chr) - 2$. 
Let $\ext'_2$ be a most parsimonious extension of $\chi'$ to $T_2$. We say that  $\ext'_2$ is \emph{left merging} if, in  $\ext'_2$, there is a simple path $P$ from some taxon in $X_{\state{A}}$ to some taxon in
$X_{\state{B}}$ such that all vertices on $P$ are allocated state $\state{B}$ by $\ext'_2$ and $P$ traverses edge $e_{\state{A}}$. We say that  $\ext'_2$ is \emph{right merging}
 if, in  $\ext'_2$, there is a simple path $P$ from some taxon in $X_{\state{A}}$ to some taxon in
$X_{\state{B}}$ such that all vertices on $P$ are allocated state $\state{B}$ by $\ext'_2$ and $P$ traverses both edges edge $e_{\state{AC}}$ and $e_{\state{C}}$. Note that  $\ext'_2$ might be left merging, right merging, both or neither. Depending on the exact combination, we use a different relabeling strategy.

\begin{figure}[H]
\centering
\pbox{\textwidth}{\includegraphics{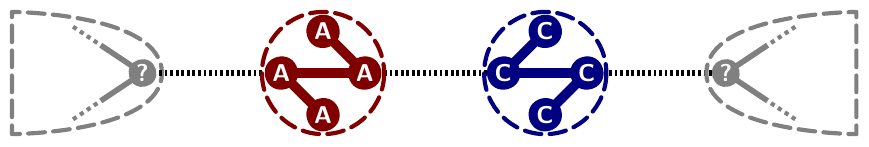}}\quad\parbox{0.3\textwidth}{\centering\small\itshape $\ext_2$ satisfies the \\ lemma requirements.}\\[2ex]
\pbox{\textwidth}{\includegraphics{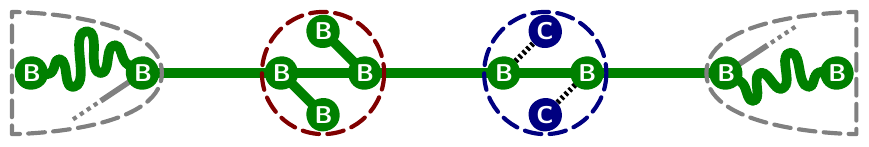}}\quad\parbox{0.3\textwidth}{\centering\small\itshape $\ext'_2$ is both left merging \\ and right merging.}
\caption{Top: the situation described in Lemma \ref{lem:neighbours}. Bottom: the
fourth case in the proof of that lemma.}
\label{fig:deg2mutation}
\end{figure}

The simplest is the case when  $\ext'_2$  is neither left merging nor right merging. In this case, consider the subgraph of $T_2$ induced
by vertices that are allocated state $\state{B}$ by $\ext'_2$. In general this subgraph might be disconnected. Delete all connected components of the subgraph that do not contain at least one taxon from $X_{\state{A}}$. Now, let $\mathcal{V}'$ be the vertices that remain. We create an extension $\exx$ of $\chr$ from $\ext'_2$ by relabeling all vertices in $\mathcal{V}'$ to state $\state{A}$, and leaving the other vertices untouched. (There is no danger that
a taxon in $X_{\state{B}}$ will be labeled with state $\state{A}$ because that would mean
$\ext'_2$ was left and/or right merging, which we exclude by assumption.) Given that
$X_{\state{A}}$ will by construction be a subset of $\mathcal{V}'$,  $\exx$ is indeed a valid extension of $\chr$. Moreover, $\Delta(T_2,\exx) = \Delta(T_2, \ext'_2)$. This is because, due to the fact that $\ext'_2$ is neither
left or right merging, the transformation of $\ext'_2$ into $\exx$ cannot create any new mutations. This then gives $\ell(T_2,\chi) \leq |\Delta(T_2,\exx)|= |\Delta(T_2,\ext'_2)|=\ell(T_2,\chi') \leq \ell(T_2,\chi)-2$, and we have our desired contradiction.

If $\ext'_2$ is left merging but not right merging, consider the subgraph of $T_2$ induced by vertices that are allocated state $\state{B}$ by $\ext'_2$. Delete edge $e_{\state{A}}$ from the subgraph. (It will definitely
be in the subgraph because $\ext'_2$ is left merging). Next delete all connected components of the subgraph that do not contain at least one taxon from $X_{\state{A}}$. As above, transform
$\ext'_2$ into $\exx$, an extension of $\chr$, by relabeling all the surviving vertices from \state{B} to \state{A}. The transformation can only increase the number of mutations by at most 1: on the edge $e_{\state{A}}$. Hence $\ell(T_2,\chi) \leq |\Delta(T_2,\exx)| \leq |\Delta(T_2,\ext'_2)|+1= \ell(T_2,\chi') + 1 \leq (\ell(T_2,\chi)-2)+1 = \ell(T_2,\chi)-1$, and we again have a contradiction.

If $\ext'_2$ is right merging but not left merging, we do exactly the same as in the previous paragraph, except that we delete $e_{\state{AC}}$ instead of $e_{\state{A}}$. This again yields the contradiction
$\ell(T_2,\chr) \leq \ell(T_2,\chr) - 1$.

The final, and most complicated case, is when $\ext'_2$ is both left merging and right merging (see figure \ref{fig:deg2mutation}, bottom subfigure). Here we convert $\ext'_2$ into $\exx$ as follows: all vertices in $\mathcal{V}_{\state{A}}$ are switched to state $\state{A}$,
and all vertices in $\mathcal{V}_{\state{C}}$ are switched to state $\state{C}$. This can create a new mutation on edge $e_{\state{A}}$. (The relabeling might cause some mutations inside $\mathcal{V}_{\state{A}}$ to disappear, which can only help us, but
for the sake of the proof we shall not assume this advantage exists). The relabeling can also create new mutations on $e_{\state{A}\state{C}}$ and $e_{\state{C}}$. However, these two mutations are compensated for
by the disappearance of at least two mutations inside $\mathcal{V}_{\state{C}}$.  The argument is as follows. Clearly, $\state{C} \neq \state{B}$ because $\state{C}$ is unique. The fact that $\ext'_2$ is right merging means
that (in $\ext'_2$) it is possible to walk along a simple path from some taxon in $X_{\state{A}}$ to some taxon in $X_{\state{B}}$, such that every vertex in the path has state $\state{B}$, and the path traverses
 $e_{\state{A}\state{C}}$ and $e_{\state{C}}$. Recall that $|X_{\state{C}}| \geq 2$ and $\state{C}$ was not ``pendant'' in  $\ext_2$ (due to the assumption that Observation \ref{obs:pendant} does not hold). Hence in $\ext'_2$ there are at least two mutations of the form $\state{B}-\state{C}$ on the set of edges whose endpoints are completely contained inside $\mathcal{V}_{\state{C}}$. It is precisely these mutations that disappear
when we completely relabel $\mathcal{V}_{\state{C}}$ to state $\state{C}$. Due to this compensation effect the total increase in the number of mutations when transforming $\ext'_2$ into $\exx$ is at most 1. This yields
the by now familiar conclusion $\ell(T_2, \chr) \leq \ell(T_2,\chr) - 1$, and thus a contradiction.
\end{proof}


\subsection{The bounding function}

In this final section we show that, whenever an optimal convex character exists with strictly more than  $7\dMP\left(T_1,T_2\right) - 5$ states, then a good pair of states will definitely exist, allowing us to reduce
the number of states in the character whilst preserving optimality and convexity. This will complete the proof of the \refbst.

In particular, we will show that at least one of the situations described in Lemma \ref{lem:neighbours}, Observation \ref{obs:pendant} and Observation \ref{obs:twounique} will hold. To begin we need an auxiliary lemma.

\begin{lem}
\label{lem:redblue}
 Let $T=(V,E)$ be a (not necessarily phylogenetic) tree in which $V$ is partitioned into a set $R$ of \emph{red} vertices and a set $B$ of \emph{blue} vertices
and all leaves of $T$ are red. If $|B| \geq 3|R|-4$, then there exist two adjacent vertices $u_1 \neq u_2$ both
of which are blue and of degree 2. 
\end{lem}
\begin{proof}
Suppose for the sake of contradiction that this is not true. Let $T$ be a counter-example: all its leaves are red, and  $|B| \geq 3|R|-4$, but the two vertices with the described property (henceforth called a ``$(u_1, u_2)$  \emph{pair}'') do not exist. Now, suppose $T$ has an internal vertex $v$ that is red. We introduce a new vertex
$v'$, attach it by an edge to $v$, colour $v'$ red and colour $v$ blue. This increases the number
of blue vertices by one and preserves the number of red vertices. Moreover, due to the fact that $v$ now has degree at least 3, this operation cannot cause a $u_1, u_2$ pair to arise. Hence,
this new tree is also a counterexample. We repeat this until we obtain a tree $T'$
whose leaves are all red and whose internal vertices are all blue. Let $R'$ and $B'$ be the set of red and blue vertices of $T'$.
By the previous argument, $|B'| \geq 3|R'|-4$. Now, if one suppresses all vertices in $T'$ of degree 2,
we obtain a tree $T''$ on $|R'|$ leaves with at most $|R'|-2$ internal vertices and at most $2|R'|-3$ edges (note that these values correspond to the binary case). We can obtain $T'$ from $T''$  by subdividing each edge of $T''$ at most once. Hence,
\begin{align*}
|B'| &\leq |R'|-2 + (2|R'| - 3)\\
& = 3|R'|-5
\end{align*}
and this yields a contradiction.
\end{proof}
Now, let $\chr, \ext_1, \ext_2, F_1, F_2, G(F_2)$ be defined as at the beginning of the previous section, and let $\chr$ use strictly more than $7\dMP - 5$ (i.e. at least $7\dMP - 4$)  states where here we write
$\dMP$ as short for $\dMP\left(T_1,T_2\right)$.  If Observation \ref{obs:pendant} or Observation \ref{obs:twounique} holds then we are done. Otherwise, consider the following: $T_1$ is convex so achieves a parsimony score exactly equal to $|\chr|-1$. $T_2$ achieves a parsimony score exactly equal to $|\chr| - 1 + \dMP$, so the homoplasy score $h$ of $T_2$ is exactly $\dMP$. Then, by
Lemma \ref{lem:states} (1st inequality) there are at least $|\chr| - \dMP \geq 6 \dMP - 4$ unique states and at most $2 \dMP$ (4th inequality) repeating components (in $F_2$). We know that, because 
Observation \ref{obs:pendant} does not hold, none of the leaves of $G(F_2)$ are unique states. In particular, all the leaves of $G(F_2)$ are repeating components. Now, if we view repeating components as ``red''
vertices in Lemma \ref{lem:redblue} and unique states as ``blue'', we need $ 6 \dMP - 4 \geq 3( 2\dMP ) - 4$ to be able to use Lemma \ref{lem:redblue}.  This holds, so we are done: in particular, Lemma \ref{lem:redblue}
shows the existence of a good pair via the situation described in Lemma \ref{lem:neighbours}.

\section{Discussion}


The bound $7\dMP - 5$ is sharp for the case $\dMP = 1$: clearly
at least 2 states are needed to achieve a distance of 1 or more.
For $\dMP \geq 2$ there is probably  room to improve the bound, and this is an interesting direction for future research. For $\dMP = 2$ a slight generalization of the arguments used in the proof of Lemma \ref{lem:neighbours}, combined with an ad-hoc case analysis can be used to easily reduce the bound from 9 to 7.  Increasingly complex arguments can be utilized to reduce this further: we conjecture that 3 states are actually sufficient when $\dMP=2$.  These arguments do not easily lead to any significant improvement in the general $7 \dMP - 5$ bound and are not included here. However, they raise the intriguing (although somewhat speculative) question of whether $\dMP + 1$ states are
always sufficient; the example given later in this section shows that they are sometimes necessary.

From an algorithmic perspective the bound has the following implications. If
$k$ is a verified upper bound on $\dMP$, then we can guarantee to find
an optimal (convex) character achieving $\dMP$ simply by guessing which of $T_1$ and
$T_2$ is convex and then looping through all at most
\[
\sum_{i=2}^{7k-5} \binom{2|X|-3}{i-1}
\]
convex characters with at most $7k-5$ states. This is because a convex character with $k$ states corresponds to a size $(k-1)$ subset of the edges in the convex tree, and an unrooted tree on $|X|$ taxa has at most $2|X|-3$ edges. Clearly, for constant $k$ this yields a running time polynomial in $|X|$. (Prior to the {\refbst} a constant upper bound of $k$ states yielded only
running times of the form $O( k^{|X|} )$: there are many more non-convex than convex characters on $k$ states.) However, the bound does \emph{not} automatically mean that questions such as ``Is $\dMP \leq t$?'' or ``Is $\dMP \geq t$?'' can be answered in polynomial
time for fixed, constant $t$. This is because in its current form the {\refbst} only holds for
\emph{optimal} characters: if we apply it to suboptimal characters we can still decrease the number of states by merging good pairs of states, but the parsimony distance achieved by the new character might \emph{increase} compared to the old character. Expressed differently, the
danger exists that for some values $d < \dMP$, all convex characters achieving parsimony distance \emph{exactly} $d$ will have a huge number of states. This means that the
obvious algorithmic stategy, of looping through all convex characters with an increasing number
of states, does not have a clear stopping strategy, even for $t$ fixed.


Finally, we remark that optimal non-convex characters might have strictly fewer states than optimal convex characters. In the proof of Lemma 3.7 of \cite{dMP-fischer} the following two trees are shown which have $\dMP = 2$:
\begin{align*}
(((((((1,2),3),4),5),6),7),8);\\
(((1,3),(2,4)),((5,7),(6,8)));
\end{align*}
(The fact that $\dMP=2$ is not proven there, but it can be easily verified computationally). The proof there shows that 2 states are sufficient to achieve this maximum if non-convex characters are allowed, but 3 if we restrict to convex characters. It is natural to ask how far apart, in general, the minimum number of required states can be. 

\bibliographystyle{unsrturl}
\bibliography{convexbound}

\end{document}